\documentclass[11pt]{article}

\usepackage{fullpage}
\usepackage{amsmath,amsthm,amssymb}
\usepackage{authblk}
\usepackage{hyperref}
\usepackage{cleveref}
\usepackage{thmtools}
\usepackage{thm-restate}
\usepackage{tikz}
\usetikzlibrary{arrows}
\usetikzlibrary{positioning, fit}
\usetikzlibrary{shapes,snakes}
\usetikzlibrary{decorations.pathmorphing}
\usetikzlibrary{calc,through,intersections}
\usetikzlibrary{patterns}
\usetikzlibrary{trees}
\usetikzlibrary{arrows.meta}
\tikzset{snake it/.style={decorate, decoration=snake}}
\tikzset{arc/.style = {->,> = latex', line width=.75pt}}
\usepackage[ruled]{algorithm2e}
\usepackage[noend]{algorithmic}

\SetAlFnt{\small}
\SetAlCapFnt{\small}
\SetAlCapNameFnt{\small}
\SetAlCapHSkip{0pt}
\IncMargin{-\parindent}
\usepackage[font=small,labelfont=bf,tableposition=top]{caption}
\usepackage[font=footnotesize]{subcaption}

\newcommand{\polylog}{{\rm polylog}}

\newtheorem{lemma}{Lemma}
\newtheorem{proposition}{Proposition}
\newtheorem{theorem}{Theorem}
\newtheorem{corollary}{Corollary}
\newtheorem{myclaim}{Property}

\newenvironment{proofof}[1]{\medskip\noindent\emph{Proof of #1. }\ignorespaces}{\hfill\qed\medskip\par\noindent\ignorespacesafterend}
\newcommand{\qedclaim}{\hfill $\diamond$ \medskip}
\newenvironment{proofclaim}{\noindent{\em Proof.}}{\qedclaim}

\begin{document}

\title{Faster approximation algorithms for computing shortest cycles on weighted graphs}
\author[1,2]{Guillaume Ducoffe}
\affil[1]{\small National Institute for Research and Development in Informatics, Romania}
\affil[2]{\small University of Bucharest, Faculty of Mathematics and Computer Science, Romania}
\date{}

\maketitle

\begin{abstract}
Given an $n$-vertex $m$-edge graph $G$ with non negative edge-weights, a {\em shortest cycle} of $G$ is one minimizing the sum of the weights on its edges.
The {\em girth} of $G$ is the weight of such a shortest cycle.
We obtain several new approximation algorithms for computing the girth of weighted graphs:
\begin{itemize}
\item
For any graph $G$ with polynomially bounded integer weights, we present a deterministic algorithm that computes, in $\tilde{\cal O}(n^{5/3}+m)$-time, a cycle of weight at most twice the girth of $G$.
This matches both the approximation factor and -- almost -- the running time of the best known subquadratic-time approximation algorithm for the girth of {\em unweighted} graphs (Roditty and Vassilevska Williams, {\it SODA'12}).
Our approach combines some new insights on the previous approximation algorithms for this problem (Lingas and Lundell, {\it IPL'09}; Roditty and Tov, {\it TALG'13}) with {\sc Hitting Set} based methods that are used for approximate distance oracles and date back from (Thorup and Zwick, {\it JACM'05}).
\item
Then, we turn our algorithm into a deterministic $(2+\varepsilon)$-approximation for graphs with arbitrary non negative edge-weights, at the price of a slightly worse running-time in $\tilde{\cal O}(n^{5/3}\polylog{(1/\varepsilon)}+m)$.
For that we introduce a novel polynomial-factor approximation of the girth, that makes more amenable the passing from the graphs with bounded integer edge-weights to the general case and is of independent interest.
\item
Finally, if we insist in removing the dependency in the number $m$ of edges, we can transform our algorithms into an $\tilde{\cal O}(n^{5/3})$-time {\em randomized} $4$-approximation for the graphs with non negative edge-weights -- assuming the adjacency lists are sorted.
Combined with the aforementioned {\sc Hitting Set} based methods, this algorithm can be derandomized, thereby yielding an $\tilde{\cal O}(n^{5/3})$-time deterministic $4$-approximation for the graphs with polynomially bounded integer weights, and an $\tilde{\cal O}(n^{5/3}\polylog{(1/\varepsilon)})$-time deterministic $(4+\varepsilon)$-approximation for the graphs with non negative edge-weights.
\end{itemize}
To the best of our knowledge, these are the first known {\em subquadratic-time} approximation algorithms for computing the girth of weighted graphs. 
\end{abstract}

\section{Introduction}\label{sec:intro}

The exciting program of ``Hardness in P'' aims at proving (under plausible complexity theoretic conjectures) the {\em exact} time-complexity of fundamental, polynomial-time solvable problems in computer science.
As for {\em graph} problems, most papers in this area either deal with the dense case ($m = \Omega(n^2)$ edges)~\cite{AGV14,RoV11,WiV10}, or the sparse case ($m = {\cal O}(n)$ edges)~\cite{AVW16,ABHV+18,BRSV+18}.
-- A finer-grained classification of such problems that depends on both $n$ and $m$ has been started only recently~\cite{AgR18,LVW18} --.
In this paper, we consider the {\sc Girth} problem on edge-weighted undirected graphs, for which almost all what is known in terms of finer-grained complexity only holds for the dense case.
We recall that the girth of a given graph $G$ is the minimum weight of a cycle in $G$ --- with the weight of a cycle being defined as the sum of the weights on its edges (see Sec.~\ref{sec:prelim} for any undefined terminology in this introduction).
Orlin and Sede\~{n}o-Noda showed that this parameter can be computed in ${\cal O}(nm)$-time~\cite{OrS17}.
For dense graphs that is in ${\cal O}(n^3)$, and Vassilevska Williams and Williams~\cite{WiV10} proved a bunch of combinatorial subcubic equivalences between {\sc Girth} and other path and matrix problems.
In particular, for every $\varepsilon > 0$, there cannot exist any {\em combinatorial} $(4/3 - \varepsilon)$-approximation algorithm for {\sc Girth} that runs in truly subcubic time unless there exists a subcubic combinatorial algorithm for multiplying two boolean matrices.
Roditty and Tov completed this above hardness result with an $\tilde{\cal O}(n^2/\varepsilon)$-time $(4/3 + \varepsilon)$-approximation algorithm~\cite{RoT13}, thereby essentially completing the picture for combinatorial algorithms on dense graphs.

However, the story does not end here for at least two reasons.
The first one is that, as the graphs considered get sparser, the complexity for computing their girth falls down to ${\cal O}(n^2)$.
In fact, when the edge-weights are integers bounded by some constant $M$, there is a non combinatorial algorithm for computing the girth of {\em any} $n$-vertex graph $G$ in time $\tilde{\cal O}(Mn^{\omega})$ where $\omega$ stands for the the exponent of square matrix multiplication over a ring~\cite{RoV11}.
It is widely believed that $\omega = 2$~\cite{Rob05}, and if true, that would imply we can compute the {\em exact} value of the girth in quasi {\em quadratic} time --- at least when edge-weights are bounded. 
So far, all the {\em approximation} algorithms for {\sc Girth} on weighted graphs run in $\tilde{\Omega}(n^2)$-time.
This leads us to the following, natural research question:
\begin{center}\it
Does there exist a subquadratic approximation algorithm for {\sc Girth} on weighted graphs?
\end{center}
In this paper, we answer to this above question in the affirmative.

    \subsection{Our contributions}\label{sec:contributions}

We present new approximation algorithms for the girth of graphs with non negative real edge-weights.
These are the first algorithms to break the quadratic barrier for this problem -- at the price of a slightly worse approximation factor compared to the state of the art~\cite{RoT13}.
Of particular importance is the subcase of graphs with positive integer weights, as we can reduce the general case to it.
Our first result is obtained for the graphs with bounded integer edge-weights.

\begin{restatable}{theorem}{mainBounded}
\label{thm:main-bounded-weight}
For every $G=(V,E,w)$ with edge-weights in $\{1, \ldots, M\}$, we can compute a deterministic $2$-approximation for {\sc Girth} in time $\tilde{\cal O}(n^{5/3}\polylog{M}+m)$.
\end{restatable}

Our starting point for Theorem~\ref{thm:main-bounded-weight} is a previous $2$-approximation algorithm from Lingas and Lundell~\cite{LiL09}, that runs in {\em quadratic} time.
Specifically, these two authors introduced an $\tilde{O}(n\log{M})$-time procedure that takes as entry a specified vertex of the graph and needs to be applied to {\em every} vertex in order to obtain the desired $2$-approximation of the girth.
Inspired by the techniques used for approximate distance oracles~\cite{ThZ05} we informally modify their algorithm as follows. We only apply their procedure to the vertices in a {\em random} subset $S$: where each vertex is present with equal probability $n^{-1/3}$ (we can derandomize our approach by using known techniques from the litterature~\cite{RTZ05}).
Furthermore, for the vertices not in $S$, we rather apply a modified version of their procedure that is restricted to a small subgraph -- induced by some ball of expected size ${\cal O}(n^{1/3})$.
A careful analysis shows this is a $2$-approximation.
%

The reason why this above approach works is that, when we run the procedure of Lingas and Lundell at some arbitrary vertex $s$, it will always detect a short cycle if there is one passing {\em close} to $s$ (but not necessarily passing through $s$ itself).
This nice property has been noticed and exploited for related algorithms on {\em unweighted} graphs~\cite{LiL09}.
However, we think we are the first to prove such a property in the weighted case.
We note that one of the two algorithms proposed by Roditty and Tov in~\cite{RoT13} also satisfies such a property.
We did not find a way to exploit their algorithm in order to improve our approximation factor.

\medskip
By using a considerably more complex machinery, we finally generalize Theorem~\ref{thm:main-bounded-weight} to the graphs with {\em arbitrary} non negative edge-weights --- up to an additional polylogarithmic dependency in some error parameter $\varepsilon$ in the complexity of the algorithm.
This is the second main result of this paper:

\begin{restatable}{theorem}{mainUnbounded}
\label{thm:main-unbounded-weight}
For every $\varepsilon > 0$ and $G=(V,E,w)$ with non negative edge-weights, we can compute a deterministic $(2+\varepsilon)$-approximation for {\sc Girth} in time $\tilde{\cal O}(n^{5/3}\polylog{1/\varepsilon}+m)$.
\end{restatable}

Theorem~\ref{thm:main-unbounded-weight} implies Theorem~\ref{thm:main-bounded-weight} for some $\varepsilon = {\cal O}(1/(Mn))$ sufficiently small.
Nevertheless, we chose to present Theorem~\ref{thm:main-bounded-weight} first as the passing from bounded integer-weights to arbitrary weights is not straightforward.
In~\cite{RoT13}, Roditty and Tov introduced a nice technique -- that we partly reuse in this paper -- in order to do just that.
However, we face several new difficulties, not encountered in~\cite{RoT13}, due to the need to perform all the intermediate operations in {\em subquadratic} time.
This leads us to an interesting side contribution of this work, namely, the design of an $n^{{\cal O}(1)}$-approximation algorithm for the girth.

\medskip
Our algorithms are subquadratic in the size of the graph, but they may be quadratic in its order $n$ if there are $m = \Theta(n^2)$ edges.
By a celebrated result of Bondy and Simonovits, any unweighted graph with $m \geq 100\ell n^{1+\frac 1 {\ell}}$ edges has girth at most $2\ell$~\cite{BoS74}, and so, we can always output a constant upper-bound on the girth of moderately dense graphs. 
It implies that the dependency on $m$ can always be removed in the running-time of approximation algorithms for the girth of {\em unweighted} graphs. 
However, elementary arguments show this additional dependency in the number of edges to be necessary for weighted graphs, at least if we do not assume any particular order on the adjacency lists.
For instance, given $G=(V,E,w)$ with integer weights, add a fresh new vertex $u \notin V$ and the edges $vu, v \in V$ with (unit) weight $w_{uv} = 1$.
Then, the girth of this new graph $G'$ is exactly $2+w_{\min}$ where $w_{\min}$ denotes the minimum edge-weight in $G$.
Since any constant-factor approximation algorithm for computing $w_{\min}$ essentially requires time $\Omega(m)$, so does any constant-factor approximation algorithm for the girth of weighted graphs.

However, one may ask what happens if we now assume {\em sorted} adjacency lists -- a very natural and common assumption in the field.
As our third main result in the paper, we show that in this case, the dependency on $m$ can be removed:

\begin{restatable}{theorem}{mainWithoutM}
\label{thm:main-without-m}
Let $G=(V,E,w)$ have sorted adjacency lists.
\begin{enumerate}
\item If all edge-weights are in $\{1, \ldots, M\}$ then, we can compute a deterministic $4$-approximation for {\sc Girth} in time $\tilde{\cal O}(n^{5/3}\polylog{M})$.
\item If all edge-weights are non negative then, we can compute a {\em randomized} $4$-approximation for {\sc Girth} in time $\tilde{\cal O}(n^{5/3})$. For every $\varepsilon > 0$, we can also compute a deterministic $(4+\varepsilon)$-approximation for {\sc Girth} in time $\tilde{\cal O}(n^{5/3}\polylog{1/\varepsilon})$. 
\end{enumerate}
\end{restatable}

We observe that even assuming sorted adjacency lists, it is not clear whether the algorithm of Theorem~\ref{thm:main-bounded-weight} can be implemented to run in  time $\tilde{\cal O}(n^{5/3}\polylog{M})$.
Indeed, this algorithm requires to build several induced subgraphs in time roughly proportional to their size, that requires a different preprocessing on the adjacency lists.
We prove that we do not need to construct these induced subgraphs {\em entirely} in order to derive a constant-factor approximation of the girth. 

Similarly, for the graphs with non negative edge-weights we need to modify the algorithm of Theorem~\ref{thm:main-unbounded-weight}.
In particular, we cannot use our polynomial-factor approximation algorithm for the girth directly, as it needs to enumerate all edges in the graph.
We overcome this difficulty through the help of a classical density result for the $C_4$-free {\em unweighted} graphs~\cite{Bro66}.

    \subsection{Related work}\label{sec:related-work}

\paragraph{Approximation algorithms.}
Itai and Rodeh were the first to study the {\sc Girth} problem for {\em unweighted} graphs~\cite{ItR78}. Among other results, they showed how to compute an additive $+1$-approximation of the girth in time ${\cal O}(n^2)$.
This was later completed by Lingas and Lundell~\cite{LiL09}, who proposed a randomized quasi $2$-approximation algorithm for this problem that runs in time ${\cal O}(n^{3/2}\sqrt{\log{n}})$.
In~\cite{RoV12}, Roditty and Vassilevska Williams presented the first deterministic approximation algorithm for the girth of unweighted graphs.
Specifically, they obtained a $2$-approximation algorithm in time $\tilde{\cal O}(n^{5/3})$, and they conjectured that there does not exist any subquadratic $(2-\varepsilon)$-approximation for {\sc Girth}.
We obtain the same approximation factor for weighted graphs, and we almost match their running time up to polylog factors and to an additional term in $\tilde{O}(m)$.
It would be interesting to know whether in our case, this dependency on $m$ can be removed while preserving the approximation factor $2$.
Very recently, new subquadratic-time approximation algorithms were proposed for {\sc Girth} in unweighted graphs (see~\cite{DKS17}). 
It is open whether one can achieve a constant-factor approximation for the girth in, say, $\tilde{\cal O}(n^{1+o(1)})$-time.

Much less is known about the girth of weighted graphs.
The first known subcubic approximation was the one of Lingas and Lundell~\cite{LiL09}, that only applies to the graphs with bounded integer edge-weights.
Their work somewhat generalizes the algorithm of Itai and Rodeh for unweighted graphs. 
The approximation factor was later improved to $4/3$ by Roditty and Tov, still for the graphs with bounded integer weights, and to $4/3 + \varepsilon$ for the graphs with {\em arbitrary weights}~\cite{RoT13}.
Our algorithms in this paper are faster than these two previous algorithms, but they use the latter as a routine to be applied on several subgraphs of sublinear size.
Therefore, the approximation factors that we obtain cannot outperform those obtained in~\cite{LiL09,RoT13}.

More recently, a breakthrough logarithmic approximation of the girth of {\em directed} weighted graphs was obtained in~\cite{PRST+18}.

\paragraph{FPT in P.}
A nascent branch of the ``Hardness in P'' program aims at obtaining faster {\em exact} algorithms for polynomial-time solvable problems assuming some parameter of the instance is bounded ({\it e.g.}, see~\cite{AVW16,GMN17}).
This has been done for {\sc Girth} in~\cite{IOO18} and (only for unweighted graphs) in~\cite{CDP18},  where the authors obtained quasi linear-time algorithms parameterized by {\em treewidth} and {\em clique-width}, respectively.
We stress that a natural parameterization of the problem by the actual value of the girth does not look that promising, given that {\sc Triangle Detection} in unweighted graphs is already a hard problem on its own.  

\paragraph{Approximate distances.}
Finally, approximation algorithms for the girth are tightly related to the computation of approximate distances in weighted graphs.
In a seminal paper~\cite{ThZ05}, Thorup and Zwick showed that we can compute in expected time ${\cal O}(n^{1/k}m)$ an approximate distance oracle: that can answer any distance query in time ${\cal O}(k)$ with a multiplicative stretch at most $2k-1$.
This has been improved in several follow-ups~\cite{BaS06,Che14,PaR14,RTZ05,Wul12}. 
However, the construction of most oracles already takes (super)quadratic time for moderately dense graphs.
Our key observation in this paper is that we do not need to construct these oracles {\em entirely} if we just want to approximate the girth.
This allows us to avoid a great deal of distance computations, and so, to lower the running time.

\subsection{Organization of the paper}\label{sec:organization}

We start gathering in Section~\ref{sec:prelim} some known results from the literature that we will use for our algorithm.
Then, in Section~\ref{sec:algo}, we give some new insights on the algorithm of Lingas and Lundell~\cite{LiL09} before presenting our main result (Theorem~\ref{thm:main-bounded-weight}).
Our algorithm is generalized to graphs with arbitrary weights in Section~\ref{sec:big-weights}.
Finally, we remove the dependency on the number of edges in the time complexity of our algorithms in Section~\ref{sec:remove-edges}.
We conclude this paper with some open perspectives (Section~\ref{sec:ccl}).

\section{Preliminaries}\label{sec:prelim}

We refer to~\cite{BoM08} for any undefined graph terminology.
Graphs in this study are finite, simple (hence, without any loop nor multiple edges), connected and edge-weighted.
Specifically, we denote a weighted graph by a triple $G=(V,E,w)$ where $w: E \to \mathbb{R}^+$ is the edge-weight function of $G$.
The weight of a subgraph $H \subseteq G$, denoted $w(H) := \sum_{e \in E(H)} w_e$, is the sum of the weights on its edges.
In particular, the girth of $G$ is the minimum weight of a cycle in $G$.
The distance $dist_G(u,v)$ between any two vertices $u,v \in V$ is the minimum weight of an $uv$-path in $G$.
By extension, for every $v \in V$ and $S \subseteq V$ we define $dist_G(v,S) := \min_{u \in S} dist_G(u,v)$.
-- We will sometimes omit the subscript if no ambiguity on the graph $G$ can occur. --
For any $v \in V$ and $r \geq 0$, we also define the ball $B_G(v,r) := \{ u \in V \mid dist_G(u,v) \leq r \}$.
Finally, an $r$-nearest set for $v$ is any $r$-set ${\cal N}_r(v)$ such that, for any $x \in {\cal N}_r(v)$ and $y \notin {\cal N}_r(v)$, we have $dist_G(v,x) \leq dist_G(v,y)$.

\smallskip
For every $v \in V$, let $N_G(v) = \{ u \in V \mid uv \in E \}$ be the (open) neighbourhood of vertex $v$.
Let $Q_v = \{ vu \mid u \in N_G(v) \}$ be totally ordered.
We call it a {\em sorted} adjacency list if edges incident to $v$ are ordered by increasing weight, {\it i.e.}, $Q_v = (vu_1,vu_2,\ldots,vu_{d_v})$ and $w_{vu_i} \leq w_{vu_{i+1}}$ for every $i < d_v$.
However, we call it an {\em ordered} adjacency list if, given some fixed total ordering $\prec$ over $V$ the neighbours of $v$ are ordered according to $\prec$ ({\it i.e.}, $u_i \prec u_{i+1}$ for every $i < d_v$).
Throughout the rest of the paper we will assume that each vertex has access to two copies of its adjacency list: one being sorted and the other being ordered.
The latter can always be ensured up to an $\tilde{O}(m)$-time preprocessing.

\subsection{The Hitting Set method}

We gather many well-known facts in the literature, that can be found, {\it e.g.}, in~\cite{RTZ05,ThZ05,ACIM99,CLRS+14}.
All these facts are combined in order to prove the following useful result for our algorithms:

\begin{proposition}\label{prop:good-set}
For any graph $G=(V,E,w)$ with {\em sorted} adjacency lists, in $\tilde{O}(n^{5/3})$-time we can compute a set $S \subseteq V$, and the open balls $B_S(v) := \{ u \in V \mid dist(v,u) < dist(v,S) \}$ for every $v \in V$, such that the following two properties hold true:
\begin{enumerate}
\item $|S| = \tilde{\cal O}(n^{2/3})$;
\item and for every $v \in V$ we have $|B_S(v)| = {\cal O}(n^{1/3})$.
\end{enumerate}
\end{proposition}

It is well-known that a set $S$ as requested by Proposition~\ref{prop:good-set} can be constructed {\em randomly} as follows: every vertex in $V$ is added in $S$ with equal probability $n^{-1/3}$~\cite{ThZ05}.
We can compute such a set $S$ deterministically by combining the following two lemmata:

\begin{lemma}[~\cite{CLRS+14}]\label{lem:nearest-set}
For every $G=(V,E,w), v \in V \ \text{and} \ r \geq 0$, we can compute an $r$-nearest set for $v$ in time ${\cal O}(r^2)$ (assuming {\em sorted} adjacency lists). Furthermore, all the vertices in this set can be labeled with their distance to $v$ within the same amount of time.
\end{lemma}

\begin{lemma}[~\cite{ACIM99}]\label{lem:greedy-hitting-set}
Given $n$ sets of size $r$ over a universe of size $s$, a set $S$ of size ${\cal O}(n/r \log{n})$ hitting all $n$ sets in at least one element can be found deterministically in time ${\cal O}(s + nr)$.
\end{lemma}

\begin{proofof}{Proposition~\ref{prop:good-set}}
Set $r = \left\lfloor n^{1/3} \right\rfloor$.
We compute an $r$-nearest set ${\cal N}_r(v)$ for every vertex $v \in V$, that can be done in total time ${\cal O}(n*n^{2/3}) = {\cal O}(n^{5/3})$ by Lemma~\ref{lem:nearest-set}.
By Lemma~\ref{lem:greedy-hitting-set} (applied for $s = n$) we can compute in time ${\cal O}(n + n*n^{1/3}) = {\cal O}(n^{4/3})$ a set $S$ of size $\tilde{O}(n^{2/3})$ that intersects all the nearest sets ${\cal N}_r(v)$.
Since, for every $v \in V$, we have $B_S(v) \subseteq {\cal N}_r(v)$ by construction, we can compute this ball in time ${\cal O}(r)$ simply by scanning the $r$-nearest set, and so, we can compute all the balls $B_S(v), v \in V$ in total time ${\cal O}(nr) = {\cal O}(n^{4/3})$.
\end{proofof}

In what follows we will not only need the balls $B_S(v)$ for every vertex $v$, but also the subgraphs these balls induce in $G$.
Next, we observe that all these subgraphs can be obtained almost for free.
Namely:

\begin{lemma}[folklore]\label{lem:comute-subgraph}
For every $G=(V,E,w)$ and $U \subseteq V$ we can compute the subgraph $G[U]$ induced by $U$ in time $\tilde{\cal O}(|U|^2)$ (assuming {\em ordered} adjacency lists).
\end{lemma}

\begin{proof}
For every $x,y \in U$ we search whether $yx$ is present in the adjacency list of $y$. Since the adjacency lists are ordered, this can be done in time ${\cal O}(\log{n})$ by using a dichotomic search. 
\end{proof}

\section{Case of graphs with bounded integer weights}\label{sec:algo}

This section is devoted to the proof of Theorem~\ref{thm:main-bounded-weight}.
We start presenting some new properties of a previous approximation algorithm for the girth of weighted graphs (Section~\ref{sec:hbd}) as we will need to use them in our own algorithm.
Then, we prove our main result for graphs with bounded integer weights in Section~\ref{sec:main}.

    \subsection{Reporting a close short cycle}\label{sec:hbd}

We propose a deeper analysis of an existing approximation algorithm for {\sc Girth} on weighted graphs~\cite{LiL09}.
Roughly, this algorithm applies a same procedure to every vertex of the graph.
In order to derive the approximation factor of their algorithm, the authors in~\cite{LiL09} were considering a run that takes as entry some vertex {\em on a shortest cycle}.
This is in contrast with the classical algorithm from Itai and Rodeh on unweighted graphs~\cite{ItR78}, that also offers provable guarantees on the length of the output assuming there is a shortest cycle passing {\em close} to the source (but not necessarily passing by this vertex); see~\cite[Lemma 2]{LiL09}.
We revisit the analysis of the algorithm in~\cite{LiL09} for weighted graphs, and we prove that this algorithm also satisfies such a ``closeness property''.

\paragraph{The {\tt HBD}-algorithm from~\cite{LiL09}.}
Given $G=(V,E,w)$, $s \in V$ and $t \geq 0$, the algorithm $HBD(G,s,t)$ is a relaxed version of Dijkstra's single-source shortest-path algorithm.
We are only interested in computing the ball of radius $t$ around $s$, and so, we stop  if there is no more unvisited vertex at a distance $\leq t$ from $s$.
Furthermore, whenever we visit a vertex $u \in B_G(s,t)$, we only relax edges $e = \{u,v\}$ such that $dist(s,u) + w_e \leq t$.
Then, a cycle is detected if we already inferred that $dist(s,v) \leq t$ ({\it i.e.}, using another neighbour of $v$ than $u$).
Overall, the algorithm stops as soon as it encounters a cycle, or all the vertices in $B_G(s,t)$ were visited. 
Assuming sorted adjacency lists, each call to this algorithm runs in $\tilde{\cal O}(n)$-time~\cite{LiL09}.

\begin{figure}[!ht]
    \centering
    \begin{subfigure}{.3\textwidth}
        \caption{\texttt{HBD}$(G,s,t)$}
        \centering
        \begin{algorithmic}[1]\small
            \FORALL{$v \in V$}
            \STATE $d(v) \gets \infty$; $\pi(v) \gets NIL$
            \ENDFOR
            \STATE $d(s) \gets 0$; $Q \gets \{s\}$
            \WHILE{$Q \neq \emptyset$}
            \STATE $u \gets \texttt{Extract-min}(Q)$
            \STATE \texttt{Controlled-Relax}$(u,t)$
            \ENDWHILE
        \end{algorithmic}
     \end{subfigure}
        \hfill
     \begin{subfigure}{.3\textwidth}
        \caption{\texttt{Controlled-Relax}$(u,t)$}
        \centering
        \begin{algorithmic}[1]\small
            \STATE $Q_u \gets$ sorted adj. list
            \STATE $uv \gets \texttt{Extract-min}(Q_u)$
            \WHILE{$d(u) + w_{uv} \leq t$}
            \STATE \texttt{RelaxOrStop}$(u,v)$
            \STATE $uv \gets \texttt{Extract-min}(Q_u)$
            \ENDWHILE
        \end{algorithmic}
     \end{subfigure}
     \begin{subfigure}{.3\textwidth}
        \caption{\texttt{RelaxOrStop}$(u,v)$}
        \centering
        \begin{algorithmic}[1]\small
            \IF{$d(v) \neq \infty$}
            \RETURN a cycle and stop
            \ELSE
            \STATE $d(v) \gets d(u) + w_{uv}$
            \STATE $Q \gets Q \cup \{v\}$
            \ENDIF
        \end{algorithmic}
     \end{subfigure}
    \label{algorithms}
\end{figure}


\begin{lemma}[~\cite{LiL09}]\label{lem:stop-output}
If \texttt{HBD}$(G,s,t)$ detects a cycle, then its weight is $\leq 2t$.
\end{lemma}

We now complete the analysis of the \texttt{HBD}-algorithm in order to derive a generalization of~\cite[Lemma 2]{LiL09} to weighted graphs.
Assuming no cycle has been detected, we first gain more insights on the structure of the ball of radius $t$ centered at $s$.
Specifically, what the following lemma just says is that the set of all edges relaxed by the algorithm induces a spanning tree of the subgraph that is induced by $B_G(s,t)$.

\begin{lemma}\label{lem:non-detect-cond}
If \texttt{HBD}$(G,s,t)$ does not detect a cycle then, for any $v \in B_G(s,t)$, there exists a unique $sv$-path of weight $\leq t$.
\end{lemma}

\begin{proof}
Suppose for the sake of contradiction that there exist $P_1,P_2$ two different $sv$-path of weight $\leq t$.
For every $i \in \{1,2\}$, let $P_i = (x_0^i = s, x_1^i, x_2^i, \ldots, x_{\ell_i}^i = v)$.
We also denote by $P_i[x_p^i,x_q^i]$ the subpath between $x_p^i$ and $x_q^i$.
Then, for every $j \in \{0, 1, 2, \ldots, \ell_i\}$, $dist(s,x_j^i) \leq w(P_i[s,x_j^i]) \leq w(P_i) \leq t$, and so (since no cycle is detected), $x_j^i$ is visited by the algorithm.
Furthermore, if $j \neq \ell_i$ then, $dist(s,x_j^i) + w_{x_j^ix_{j+1}^i} \leq w(P_i[s,x_{j+1}^i]) \leq w(P_i) \leq t$, and so, the edge $x_j^ix_{j+1}^i$ is relaxed.
Overall, all the edges of $E(P_1) \cup E(P_2)$ are relaxed, and we claim that it contradicts our assumption that no cycle has been detected.
Indeed, since $P_1,P_2$ are different and they have the same ends, there exists a cycle $C$ such that $E(C) \subseteq E(P_1) \cup E(P_2)$ ({\it e.g.}, see~\cite[Lemma 2.5.]{RoV12}).
Let $e = xy \in E(C)$ be the last edge relaxed on the cycle.
As $x$ and $y$ are each incident to other edges of $E(C)$ that are already relaxed, we have at this step $d(x) \neq \infty$ and $d(y) \neq \infty$, thereby proving the claim.
\end{proof}

Based on Lemma~\ref{lem:non-detect-cond}, we state some bounds on the weight of the cycle detected using {\tt HBD}. 

\begin{corollary}\label{cor:stop-hbd-1}
Given $G=(V,E,w)$, let $s \in V$ and let $C$ be a cycle.
The minimum $t_0$ such that \texttt{HBD}$(G,s,t_0)$ detects a cycle satisfies $t_0 \leq dist(s,C) + w(C)$.
\end{corollary}

\begin{proof}
Fix $x,y \in V(C)$ and partition the cycle $C$ into the two $xy$-paths $P_1,P_2$.
By the contrapositive of Lemma~\ref{lem:non-detect-cond}, a cycle is detected if $t \geq \min\{dist(s,x),dist(s,y)\} + \max\{w(P_1),w(P_2)\}$.
Indeed, in this situation there exist two different paths of length $\leq t$ between $s$ and one of $x$ or $y$.
In particular, set $dist(s,x) = dist(s,C)$ and let $y \in V(C) \setminus \{x\}$ be arbitrary.
We have $t_0 \leq dist(s,C) + \max\{w(P_1),w(P_2)\} \leq dist(s,C) + w(C)$.
\end{proof}

We end up improving the bound of Corollary~\ref{cor:stop-hbd-1} in some particular cases of interest. The following refined upper-bound will play a key role in the analysis of our algorithm in the next sections.

\begin{corollary}\label{cor:stop-hbd-2}
Given $G=(V,E,w)$, let $s \in V$ and let $C$ be a cycle.
Assume the existence of a vertex $x \in V(C)$ such that $\max_{v \in V(C)} dist_C(x,v) \geq dist_G(s,x) > 0$. Then, the minimum $t_0$ such that \texttt{HBD}$(G,s,t_0)$ detects a cycle satisfies $t_0 \leq w(C)$.
\end{corollary}

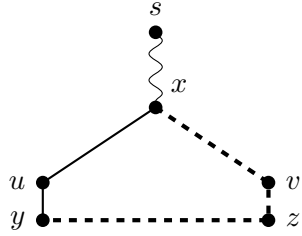
\begin{center}
\begin{tikzpicture}
\draw node[scale=.5,fill,circle,label=90:$s$] at (0,1) {};
\draw node[scale=.5,fill,circle,label=45:$x$] at (0,0) {};
\draw node[scale=.5,fill,circle,label=180:$u$] at (-1.5,-1) {};
\draw node[scale=.5,fill,circle,label=180:$y$] at (-1.5,-1.5) {};
\draw node[scale=.5,fill,circle,label=0:$v$] at (1.5,-1) {};
\draw node[scale=.5,fill,circle,label=0:$z$] at (1.5,-1.5) {};
\draw[thick] (0,0) -- (-1.5,-1) -- (-1.5,-1.5);
\draw[snake it] (0,1) -- (0,0);
\draw[dashed,ultra thick] (0,0) -- (1.5,-1) -- (1.5,-1.5) -- (-1.5,-1.5);
\end{tikzpicture}
\end{center}

\begin{proof}
Let $B_x = \{ v \in V(C) \mid dist_C(x,v) < dist_G(x,s) \}$, and let $P_x$ be a shortest-path tree of $C[B_x]$ rooted at $x$.
By construction, $P_x$ is a path such that $w(P_x) < 2 \cdot dist_G(x,s)$.
Furthermore since we assume $\max_{v \in V(C)} dist_C(x,v) \geq dist_G(x,s)$, $B_x \neq V(C)$.
Hence there exist $uy,vz \in E(C)$ such that $u,v \in B_x$ but $y,z \notin B_x$ (possibly, $y=z$ or $u=v=x$, but not both at the same time).
W.l.o.g. $dist_C(x,y) \leq dist_C(x,z)$.
We can bipartition $E(C)$ in two edge-disjoint $xy$-paths $P_1$ and $P_2$, with $P_1$ being the $xy$-subpath passing by $vz$ (and so, $P_2$ is the other $xy$-subpath passing by $uy$).
Note that it implies $dist_C(x,y) = w(P_2) \leq w(C)/2$.
Then, by Lemma~\ref{lem:non-detect-cond} we have $t_0 \leq dist_G(s,x) + \max\{w(P_1),w(P_2)\} = dist_G(s,x) + w(P_1) \leq dist_C(x,y) + w(P_1) = w(P_2) + w(P_1) = w(C)$.
\end{proof}

    \subsection{Subquadratic-time approximation}\label{sec:main}

We are now ready to prove our first main result in this section.

\mainBounded*

\begin{proof}
We analyse the following \texttt{Subquadratic-Approx} algorithm:

\begin{figure}[!ht]
    \centering
    \begin{subfigure}{.4\textwidth}
        \caption{\texttt{Approx-Short-Close-Cycle}$(G,s,M)$}
        \centering
        \begin{algorithmic}[1]\small
            \STATE Find the minimum $t \in [3; M \cdot |V(G) | ]$ such that: ${\tt HBD}(G,s,t)$ detects a cycle.
            \STATE Let $C_s$ be the shortest cycle we so computed.
             \RETURN $C_s$.
        \end{algorithmic}
     \end{subfigure}
     \begin{subfigure}{.45\textwidth}
        \caption{\texttt{Subquadratic-Approx}$(G,M)$}
        \centering
        \begin{algorithmic}[1]\small
            \STATE Let $S$ and $(B_S(v))_{v \in V}$ be as in Proposition~\ref{prop:good-set}.
            \FORALL{$s \in S$}
            \STATE $C_s \gets {\tt Approx-Short-Close-Cycle}(G,s,M)$
            \ENDFOR
            \STATE
            \FORALL{$v \notin S$}
            \STATE Let $G'_v$ be induced by $B_S(v)$.
            \STATE $C_v \gets {\tt Approx-Short-Close-Cycle}(G'_v,v,M)$
            \ENDFOR
            \STATE
            \RETURN a shortest cycle in $\{ C_v \mid v \in V \}$.
        \end{algorithmic}
     \end{subfigure}
\end{figure}
The algorithm starts precomputing a set $S \subseteq V$ and the open balls $(B_S(v))_{v \in S}$ as described in Proposition~\ref{prop:good-set}. 
This takes time $\tilde{\cal O}(n^{5/3})$, plus an additional preprocessing time in $\tilde{O}(m)$ for sorting the adjacency lists.
Then, we process the vertices in $S$ and those in $V \setminus S$ separately:
\begin{itemize}
\item
For every $s \in S$, we compute the smallest $t_s \in [3;Mn]$ such that ${\tt HBD}(G,s,t_s)$ detects a cycle by using a dichotomic search (procedure {\tt Approx-Short-Close-Cycle}$(G,s,M)$).
We store the cycle $C_s$ outputted by ${\tt HBD}(G,s,t_s)$.
Since each test we perform during the dichotomic search consists in a call to the {\tt HBD}-algorithm, this takes time $\tilde{\cal O}(n\log{M})$ per vertex in $S$, and so, $\tilde{\cal O}(n|S|\log{M}) = \tilde{\cal O}(n^{5/3}\log{M})$ in total.


\item We now consider the vertices $v \in V \setminus S$ sequentially.
Let $G'_v$ be the subgraph of $G$ induced by the open ball $B_S(v)$.
By Lemma~\ref{lem:comute-subgraph}, this subgraph can be computed in time $\tilde{\cal O}(|B_S(v)|^2) = \tilde{\cal O}(n^{2/3})$ -- assuming a preprocessing of the graph in time ${\cal O}(m)$ for ordering the adjacency lists.
We apply the same procedure as for the vertices in $S$ but, we restrict ourselves to the ball $B_S(v)$.
That is, we call {\tt Approx-Short-Close-Cycle}$(G_v',v,M)$, and we denote by $C_v$ the cycle outputted by this algorithm.
Since we restrict ourselves to a subgraph of order ${\cal O}(n^{1/3})$, this takes total time $\tilde{\cal O}(n*(n^{2/3}+n^{1/3}\log{M})) = \tilde{\cal O}(n^{5/3}\log{M})$.
%
%
\end{itemize}

Let $C \in \{ C_v \mid v \in V \}$ be of minimum weight.
We claim that $w(C)$ is a $2$-approximation of the girth of $G$, that will end proving the theorem.
In order to prove this claim, we apply the following case analysis to some arbitrary shortest cycle $C_0$ of $G$.
\begin{itemize}
\item
If $V(C_0) \cap S \neq \emptyset$ then, let $C_S$ be a shortest cycle among $\{ C_s \mid s \in S \}$.
We prove as a subclaim that $w(C_S)$ is at most twice the weight of a shortest cycle intersecting $S$.
In order to prove this subclaim, it suffices to prove that for every $s \in S$, we compute a cycle $C_s$ of weight no more than twice the weight of a shortest cycle passing by $s$.
By Corollary~\ref{cor:stop-hbd-1}, if $t_s$ is the smallest $t$ such that $HBD(G,s,t)$ detects a cycle then, a shortest cycle $C$ passing by $s$ must have weight $\geq t_s$.
Furthermore, by Lemma~\ref{lem:stop-output} we get $w(C_s) \leq 2t_s$, thereby proving the subclaim. 
Thus, $w(C_S) \leq 2w(C_0)$ if $V(C_0) \cap S \neq \emptyset$.

\item
From now on we assume $V(C_0) \cap S = \emptyset$.
\smallskip
Let $v \in V(C_0)$ be arbitrary. There are two subcases:
\begin{itemize}
\item If $V(C_0) \subseteq B_S(v)$ then, $C_0$ is also a cycle in $G'_v$.
Moreover by Corollary~\ref{cor:stop-hbd-1} applied for $dist(v,C_0) = 0$, the smallest $t_v$ such that {\tt HBD}$(G'_v,v,t_v)$ detects a cycle satisfies $t_v \leq w(C_0)$.
By Lemma~\ref{lem:stop-output}, $w(C) \leq w(C_v) \leq 2w(C_0)$.
\item Otherwise $V(C_0) \not\subseteq B_S(v)$. 
This implies that we have: $$\max_{u \in V(C_0)} dist_{C_0}(u,v) \geq \max_{u \in V(C_0)} dist_{G}(u,v) \geq dist_G(v,S) > 0.$$
Furthermore, let $s \in S$ minimize $dist_G(s,v)$.
Then, by Corollary~\ref{cor:stop-hbd-2}, the smallest $t_s$ such that ${\tt HBD}(G,s,t_s)$ detects a cycle satisfies $t_s \leq w(C_0)$.
As a result, by Lemma~\ref{lem:stop-output} $w(C) \leq w(C_s) \leq 2w(C_0)$.
\end{itemize}
\end{itemize}
Summarizing, $w(C) \leq 2w(C_0)$ in all the cases.

\end{proof}

%

\section{Generalization to unbounded weights}\label{sec:big-weights}

In~\cite{RoT13}, Roditty and Tov first proposed a $4/3$-approximation for {\sc Girth} in graphs with bounded integer weights.
Then, they showed how to turn their algorithm into a $(4/3+\varepsilon)$-approximation for graphs with arbitrary non negative real weights.
We build on their approach in order to derive  Theorem~\ref{thm:main-unbounded-weight}.
In particular, we will use their main result as a subroutine:

\begin{theorem}[~\cite{RoT13}]\label{thm:previous-approx}
For every $G=(V,E,w)$ with arbitrary non negative edge-weights, we can compute a $(4/3+\varepsilon)$-approximation for {\sc Girth} in time $\tilde{\cal O}(n^2/\varepsilon)$.
\end{theorem}

The remaining of this section is divided into three parts.
In Section~\ref{sec:rk-girth}, we reduce the general problem to the subcase of graphs with {\em positive integer} weights.
Although strictly speaking, this part is not necessary in order to prove our results, it helps in simplifying some arguments and may be of independent interest for the further investigations on the girth.
Furthermore, our reduction from non negative weights to positive weights is a gentle introduction to the technique that we exploit in the second part of this section for the design of a {\em polynomial-factor} approximation of the girth in subquadratic time.
This part is new compared to~\cite{RoT13} and the techniques used are interesting in their own right.
Then, based on a clever technique from~\cite{RoT13}, we end up refining this rough estimate of the girth until we obtain a constant-factor approximation.

\subsection{Reduction to positive integer weights}\label{sec:rk-girth}

In what follows are some reductions from the graphs with arbitrary non negative real weights to the graphs with positive integer (unbounded) weights.
First we assume all the weights to be positive (we will show how to deal with the edges of weight $0$ at the end of this section).
Recall that an $(\alpha,\beta)$-approximation algorithm for the girth is one returning a cycle of weight no more than $\alpha \cdot g^* + \beta$ for any graph with girth $\leq g^*$.
Our first observation is that since we are dealing with weighted graphs, given an $(\alpha,\beta)$-approximation algorithm for the girth we can always scale the weights in order to obtain an $(\alpha + o(1))$-approximation.

\begin{lemma}\label{lem:remove-constant}
Assume there exists an $T(n,m)$-time $(\alpha,\beta)$-approximation algorithm for {\sc Girth} for the graphs with positive real edge-weights, where $\alpha,\beta$ do not depend on the weights and $T(n,m) = \Omega(m)$.
Then, for every $\varepsilon > 0$, there also exists an ${\cal O}(T(n,m)\log{1/\varepsilon})$-time $(\alpha+\varepsilon)$-approximation algorithm for the girth.
\end{lemma} 

\begin{proof}
Given $G=(V,E,w)$, let $w_{\min} := \min \{ w_e \mid e \in E \}$ (computable in time ${\cal O}(m)$).
We multiply all the edge-weights by $\beta/(3\varepsilon w_{\min})$.
In doing so, the girth must be at least $g(G) \geq 3 \times \beta/(3\varepsilon) = (1/\varepsilon)\beta$, and so, we have $\alpha \cdot g(G) + \beta \leq (\alpha + \varepsilon)g(G)$.
\end{proof}

In~\cite{DKS17}, Dahlgaard et al. presented various subquadratic $(1 + \varepsilon, {\cal O}(1))$-approximation algorithms for the girth of {\em unweighted} graphs, for every $\varepsilon > 0$. 
We recall that $4/3$ is, under some complexity theoretic assumptions, the best possible approximation factor that one can get using a subcubic combinatorial algorithm~\cite{RoT13}.
Therefore by Lemma~\ref{lem:remove-constant} such combinatorial algorithms are unlikely to generalize to weighted graphs.

Another more straightforward consequence of Lemma~\ref{lem:remove-constant} is that we can always assume the weights to be (possibly very large) integers.
This might be useful for the further investigations on the girth.

\begin{corollary}\label{cor:integer-weights}
Assume there exists an $T(n,m)$-time $\alpha$-approximation algorithm for {\sc Girth} for the graphs with positive integer edge-weights, where $T(n,m) = \Omega(m)$.
Then, for every $\varepsilon > 0$, there exists an ${\cal O}(T(n,m)\log{1/\varepsilon})$-time $(\alpha+\varepsilon)$-approximation algorithm for {\sc Girth} for the graphs with positive real edge-weights.
\end{corollary}

\begin{proof}
Replace every edge-weight $w_e, e \in E$ by $\left\lceil w_e \right\rceil$.
Since $w_e \leq \left\lceil w_e \right\rceil \leq w_e + 1$, the girth is increased by at most $n$.
In particular, using an $\alpha$-approximation algorithm for the graphs with integer weights, we output a cycle of weight at most $\alpha \cdot g(G) + \alpha n$, with $g(G)$ being the girth of $G$.
We can use Lemma~\ref{lem:remove-constant} in order to turn this result into an $(\alpha+\varepsilon)$-approximation for every $\varepsilon > 0$.  
\end{proof}

Finally, we reduce the general case of graphs with non negative weights to the subcase of graphs with positive weights (and so, also to the subcase of graphs with positive integer weights).
The technique that we use for that will be also used, in a more complicated way, in order to prove Proposition~\ref{prop:linear-approx} in the next section.

\begin{lemma}\label{remove-zero-weights}
Assume there exists an $T(n,m)$-time $\alpha$-approximation algorithm for {\sc Girth} for the graphs with positive edge-weights, where $T(n,m) = \Omega(m)$.
Then, there also exists an ${\cal O}(T(n,m))$-time $\alpha$-approximation algorithm for {\sc Girth} for the graphs with non negative edge-weights.
\end{lemma}

\begin{proof}
Let $G=(V,E,w)$ with non negative real edge-weights, and let $E_0 = \{ e \in E \mid w_e = 0 \}$.
If the (unweighted) graph $G_0 := (V,E_0)$ contains a cycle $C$, then $w(C) = 0$ is minimized and we can output $C$.
Otherwise, $G_0$ is a forest, and let $V_1,V_2,\ldots,V_p$ be its connected components.
We consider the following three different cases:
\begin{enumerate}
\item First we scan all the edges in $E \setminus E_0$ in order to find, for any component $V_i$, the minimum weight of an edge with its two ends in $V_i$ (if any).
It takes time ${\cal O}(m)$.
Furthermore, note that given $e \in E \setminus E_0$ with its two ends in $V_i$, there is a cycle of weight exactly $w_e$.
Let $C_0$ be the minimum-weight cycle we so computed during this step.

\item
In the same way, we can easily find a cycle $C_1$ of minimum weight that only intersects two components: for that, we scan all the edges in $E \setminus E_0$ in order to find, for any two distinct $V_i,V_j$, the at most two edges of minimum-weight with one end in $V_i$ and the other end in $V_j$.
It takes time ${\cal O}(m)$.

\item
We are now left with approximating the short cycles that intersect at least three components of $G_0$.
For that, let $G'$ be obtained from $G$ by contracting each component $V_j$ into one vertex; for every distinct $V_i,V_j$, if there exists an edge $e \in E \cap (V_i \times V_j)$ then, we choose $e$ minimizing $w_e$ and we set $w'_{v_iv_j} = w_e$ in $G'$.
We observe that if there exists a shortest cycle $C$ of $G$ that intersects at least three components $V_j$ then, $C \cap V_i$ is either empty or induces a path in $G_0[V_i]$ for every $i$ (otherwise, there would be two vertices in $V_i$ that are connected in $C$ by two subpaths of positive weight, and so, we could obtain a cycle of smaller weight than $C$ by replacing any of these subpaths with any path in $G_0[V_i]$).
In particular, it corresponds to $C$ a cycle $C'$ in $G'$ such that $w(C') \leq w(C)$.
Conversely, to any cycle $C'$ in $G$, it corresponds a cycle $C$ in $G$ such that $w(C) \leq w(C')$ (obtained by uncontracting the connected components of $G_0$).
Let $C_2'$ be a cycle of weight at most $\alpha$ times the girth of $G'$, and let $C_2$ be a corresponding cycle in $G$.
\end{enumerate}
Then, by outputting a cycle among $C_0,C_1,C_2$ that is of minimum weight, one obtains an $\alpha$-approximation for the girth of $G$.
\end{proof}

\subsection{A polynomial-factor approximation.}

We now obtain an approximation of the girth that only depends on the order of the graph.
Since we consider weighted graphs, this is already a non trivial and interesting achievement.

\begin{proposition}\label{prop:linear-approx}
For every $G=(V,E,w)$ with arbitrary positive edge-weights, we can compute an $\tilde{\cal O}(n^{2/3})$-approximation for {\sc Girth} in time $\tilde{\cal O}(n^{5/3} + m)$.
\end{proposition}

\begin{proof}
Let $S$ be as in Proposition~\ref{prop:good-set}.
We show a significantly more elaborate method that uses $S$ in order to approximate the girth.
We divide this method into five main steps.
\paragraph{Step 1: check the small balls.}
For every $v \notin S$, let $G'_v$ be the subgraph induced by the open ball $B_S(v)$.
As before, we first estimate the girth of $G'_v$.
Since this subgraph has order ${\cal O}(n^{1/3})$, by Theorem~\ref{thm:previous-approx} we can compute a constant-factor approximation for its girth in time $\tilde{\cal O}(n^{2/3})$ (say, a $2$-approximation).
Overall, this step takes total time $\tilde{\cal O}(n^{5/3})$.
Furthermore, after completing this step the following property (also used in Theorem~\ref{thm:main-bounded-weight}) becomes true:

\begin{myclaim}\label{claim:step-1}
Let $C_v$ be a shortest cycle passing through $v$.
If $w(C_v) < 2 \cdot dist(v,S)$ then, we computed a cycle of weight $\leq 2w(C_v)$.
\end{myclaim}

\begin{proofclaim}
This is trivial if $v \in S$.
Otherwise by the hypothesis, $V(C_v) \subseteq B_G(v,w(C_v)/2) \subseteq B_S(v) = V(G_v')$, and so, we outputted a cycle of weight $\leq 2w(C_v)$ for this subgraph.
\end{proofclaim}

We will use Property~\ref{claim:step-1} repeatedly throughout the next steps.
\paragraph{Step 2: partitioning into (shortest path) subtrees.}
Intuitively, what we try to do next is to approximate the weight of a shortest cycle passing close to $S$.
The difference with Theorem~\ref{thm:main-bounded-weight} is that we cannot use directly the algorithm of Roditty and Tov for that.
Indeed, their algorithm has some global steps ({\it e.g.}, the approximate computation of the girth of some sparse spanner) that we currently do not know how to do in subquadratic time.
So, we need to find some new techniques. 
Specifically, we partition the vertex-set $V$ into shortest-path subtrees $(T_s)_{s \in S}$ such that, for every $s \in S$ and $v \in V(T_s)$ we have $dist(v,s) = dist(v,S)$.
As noted, {\it e.g.}, in~\cite{ThZ05}, a simple way to do that is to add a dummy vertex $x_S \notin V$, edges $sx_S$ for every $s \in S$ with weight $0$, then to compute a shortest-path tree rooted at $x_S$ in time $\tilde{\cal O}(m)$.
See Fig.~\ref{fig:partitioning-step} for an example.
In what follows, we show how to use this tree structure in order to compute short cycles.

\begin{figure}[h!]
\begin{subfigure}{.45\textwidth}
\begin{center}
\begin{tikzpicture}

\draw (-2,0) -- (-1,1) -- (0,1) -- (1,1) -- (2,0) -- (1,-1) -- (0,-1) -- (-1,-1) -- (-2,0);
\draw (-1,-1) -- (-1,1) -- (0,-1) -- (1,1) -- (1,-1);
\draw (0,-1) -- (0,1);

\draw node[scale=.5,circle,fill] at (-2,0) {};
\draw[red] node[scale=1,rectangle,fill] at (-1,1) {};
\draw node[scale=.5,circle,fill] at (-1,-1) {};
\draw node[scale=.5,circle,fill] at (0,1) {};
\draw node[scale=.5,circle,fill] at (0,-1) {};
\draw[red] node[scale=1,rectangle,fill] at (1,1) {};
\draw node[scale=.5,circle,fill] at (1,-1) {};
\draw node[scale=.5,circle,fill] at (2,0) {};

\node at (-1.5,0.8) {$4$};
\node at (-1.5,-0.8) {$7$};
\node at (1.5,0.8) {$9$};
\node at (1.5,-0.8) {$3$};
\node at (-.5,1.3) {$3$};
\node at (.5,1.3) {$3$};
\node at (-.5,-1.3) {$8$};
\node at (.5,-1.3) {$3$};
\node at (-1.3,0) {$2$};
\node at (-.3,.5) {$7$};
\node at (1.3,0) {$7$};
\node at (-.6,-.4) {$9$};
\node at (.6,-.4) {$2$};

\end{tikzpicture}
\end{center}
\end{subfigure}\hfill
\begin{subfigure}{.45\textwidth}
\begin{center}
\begin{tikzpicture}

\draw[red,ultra thick,dashed] (-2,0) -- (-1,1) -- (0,1);
\draw[red,ultra thick,dashed] (-1,1) -- (-1,-1);
\draw[red,ultra thick,dashed] (1,1) -- (0,-1) -- (1,-1) -- (2,0);

\draw node[scale=.5,circle,fill] at (-2,0) {};
\draw[red] node[scale=1,rectangle,fill] at (-1,1) {};
\draw node[scale=.5,circle,fill] at (-1,-1) {};
\draw node[scale=.5,circle,fill] at (0,1) {};
\draw node[scale=.5,circle,fill] at (0,-1) {};
\draw[red] node[scale=1,rectangle,fill] at (1,1) {};
\draw node[scale=.5,circle,fill] at (1,-1) {};
\draw node[scale=.5,circle,fill] at (2,0) {};

\node at (-1.5,0.8) {$4$};
\node at (-.5,1.3) {$3$};
\node at (.5,-1.3) {$3$};
\node at (-1.3,0) {$2$};
\node at (.6,-.4) {$2$};
\node at (1.5,-0.8) {$3$};

\end{tikzpicture}
\end{center}
\end{subfigure}
\caption{An example of Step 2. The two vertices in $S$ are drawn as rectangles.}
\label{fig:partitioning-step}
\end{figure}
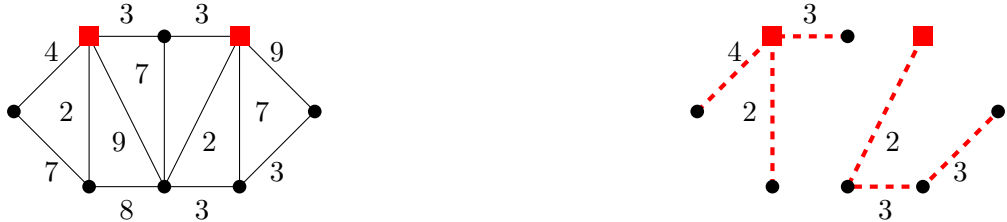

\paragraph{Step 3: finding short cycles in a subtree.}
Let $s \in S$ be fixed.
Informally, we try to estimate the weight of a shortest cycle in $V(T_s)$.
Note that every such a cycle has an edge that is not contained in $T_s$.
So, we consider all the edges $e = uv$ such that $u,v \in V(T_s)$ but $e \notin E(T_s)$.
Adding this edge in $T_s$ closes a cycle.
Let $C_{e,s}$ be an (unknown) shortest cycle passing by $e$ and contained in $V(T_s)$.
We output $dist(s,u) + w_e + dist(v,s)$ as a rough estimate of $w(C_{e,s})$.
Indeed, the latter is a straightforward upper-bound on $w(C_{e,s})$, and this bound is reached if $s \in \{u,v\}$.
Overall, this step takes total time ${\cal O}(m)$.

\begin{myclaim}\label{claim:step-3}
Let $C^*_s$ be a shortest cycle contained in $V(T_s)$.
After Steps 1-3, we computed a cycle of weight $\leq 2 w(C^*_s)$.
\end{myclaim}

\begin{proofclaim}
Let $e = uv \in E(C^*_s) \setminus E(T_s)$.
We also define $C_u,C_v$ shortest cycles passing by $u$ and $v$, respectively.
By Property~\ref{claim:step-1}, either we computed at Step 1 a cycle of weight $\leq 2 w(C_u) \leq 2 w(C^*_s)$ (a cycle of weight $\leq 2 w(C_v)) \leq 2 w(C^*_s)$, resp.), or we know for sure that $w(C^*_s) \geq w(C_u) \geq 2 dist(u,S) = 2 dist(u,s)$ ($w(C^*_s) \geq w(C_v) \geq 2 dist(v,S) = 2 dist(v,s)$, resp.). 
In the latter case, $dist(s,u) + w_e + dist(v,s)$ is a $2$-approximation of $w(C^*_s)$.
\end{proofclaim}

The proof for the next step is quite similar.
The approximation factor we get is slightly worse, but it is still a constant.

\paragraph{Step 4: finding short cycles in two subtrees.}
We now want to estimate the weight of a shortest cycle in $V(T_s) \cup V(T_{s'})$, for some distinct $s,s' \in S$.
There are three cases: either (a) such a cycle is fully contained in $V(T_s)$; or (b) it is fully contained in $V(T_{s'})$; or (c) it must contain two edges $e,e'$ with an end in $V(T_s)$ and the other end in $V(T_{s'})$.
Cases (a) and (b) have already been considered at Step 3.
So, we only consider Case (c), and we proceed as follows:
\begin{enumerate}
\item We scan all the edges $e = uv \in E$ such that $u$ and $v$ are not in a same subtree.
Let $s_u,s_v \in S$ such that $u \in V(T_{s_u}), v \in V(T_{s_v})$.
We set $\ell(e) = dist(s_u,u) + w_e + dist(v,s_v)$.
\item Group all these above edges with their two ends in the same two subtrees. 
It takes time ${\cal O}(m+|S|) = {\cal O}(m + n^{2/3})$ by using, say, a linear-time sorting algorithm.
\item Finally, for every distinct $s,s' \in S$, let $E(s,s')$ contain all the edges with one end in $T_s$ and the other end in $T_{s'}$.
If $|E(s,s')| \geq 2$ then, we pick $e,e'$ minimizing $\ell(\cdot)$ and we output $\ell(e)+\ell(e')$.
Overall, since the sets $E(s,s')$ partition the edges of $G$, this last phase also takes time ${\cal O}(m)$.
\end{enumerate}
To prove correctness of this step, let $s,s' \in S$ be distinct and fixed.
We prove the following result:

\begin{myclaim}\label{claim:step-4}
Let $C^*_{s,s'}$ be a shortest cycle contained in $V(T_s) \cup V(T_{s'})$.
After Steps 1-4, we computed a cycle of weight $\leq 3 w(C_{s,s'}^*)$.
\end{myclaim}

\begin{proofclaim}
If either $V(C^*_{s,s'}) \subseteq V(T_s)$ or $V(C^*_{s,s'}) \subseteq V(T_{s'})$ then, we are done by Property~\ref{claim:step-3}.
Otherwise, let $e = uv', e' = vu' \in E(s,s') \cap E(C^*_{s,s'})$ such that $u,v \in V(T_s)$ and $u',v' \in V(T_{s'})$.
Choosing the $uv$-path in $T_s$ and the $u'v'$-path in $T_{s'}$, one obtains that $w(C^*_{s,s'}) \leq \ell(e) + \ell(e')$.
Furthermore, let $C_u$ be a shortest cycle passing by $u$.
By Property~\ref{claim:step-1}, either we computed at Step 1 a cycle of weight $\leq 2 w(C_u) \leq 2 w(C^*_{s,s'})$, or we know for sure that $w(C^*_{s,s'}) \geq w(C_u) \geq 2 dist(u,S) = 2 dist(u,s)$.
We obtain similar results for $v,u',v'$.
As a result, and unless we found a better estimate of the girth during Step 1, we have $\ell(e) + \ell(e') \leq 2w(C^*_{s,s'}) + w_e + w_{e'} \leq 3 w(C^*_{s,s'})$. 
\end{proofclaim}

\paragraph{Step 5: the general case.}
We end up defining a weighted graph $H_S = (S, E_S, w^S)$, where:
$$ E_S = \{ ss' \mid E(s,s') \neq \emptyset \}, $$
and for every $ss' \in E_S$:
$$w^S_{ss'} = \min\limits_{e \in E(s,s')} \ell(e) = \min\limits_{uv \in E(s,s')} dist(s,u) +w_{uv} + dist(v,s').$$
We can construct $H_S$ simply by scanning all the sets $E(s,s')$ (computed during Step 4).
Overall, since the sets $E(s,s')$ partition the edges of $G$, this takes total time ${\cal O}(m + |S|) = {\cal O}(m + n^{2/3})$. 
Furthermore, by Theorem~\ref{thm:previous-approx} we can compute a constant-factor approximation of the girth of $H_S$ in time $\tilde{\cal O}(|S|^2) = \tilde{\cal O}(n^{4/3})$.

\smallskip
The graph $H_S$ is {\em not} a subgraph of $G$.
However, given a cycle $C_H$ for $H_S$, we can compute a cycle $C_H^*$ of $G$ as follows.
For every $s \in V(C_H)$ let $s',s'' \in V(C_H)$ be its two neighbours.
By construction, there exist $e = uv \in E(s',s)$ and $e' = xy \in E(s,s'')$ such that the edges $ss'$ and $ss''$ in $H_S$ have weights $dist(s',u) + w_e + dist(v,s)$ and $dist(s,x) + w_{e'} + dist(y,s'')$, respectively.
-- We may assume the edges $e,e'$ to be stored in $H_S$ so that $s',s''$ will choose the same common edge with $s$. --
Then, we replace $s$ by the $vx$-path in $T_s$.
It is important to notice that, by construction, we have $w(C_{H}^*) \leq w(C_H)$.
In particular, we can apply this above transformation to the (approximately shortest) cycle of $H_S$ that has been outputted by the algorithm of Roditty and Tov (Theorem~\ref{thm:previous-approx}). 

\medskip
Overall, let $C_{\min}$ be a shortest cycle computed by the algorithm above ({\it i.e.}, after Steps 1-5).
In order to finish the proof, we need to show that $w(C_{\min})$ is an $\tilde{\cal O}(n^{2/3})$-approximation of the girth of $G$.
By Properties~\ref{claim:step-3} and~\ref{claim:step-4}, this is the case if there exists a shortest cycle intersecting at most two subtrees $T_s, s \in S$.
From now on assume that any shortest cycle $C_0$ of $G$ intersects at least three subtrees $T_s$.
Write $C_0=(v_0,v_1,\ldots,v_{p-1},v_0)$ and assume w.l.o.g. $v_0,v_{p-1}$ are not contained into the same subtree $T_s$.
We partition the $v_i$'s into the maximal subpaths $P_0,P_1,\ldots,P_{q-1}, \ q \leq p$ that are contained into the vertex-set of a same subtree $T_s$ (in particular, $v_0 \in V(P_0)$ and $v_{p-1} \in V(P_{q-1})$).
Furthermore for every $j \in \{0,1,\ldots,q-1\}$ let $s_j \in S$ be such $V(P_j) \subseteq V(T_{s_j})$, and let $i_j$ be the largest index such that $v_{i_j} \in V(P_j)$.
For instance, $i_{q-1} = p-1$ by construction.
Since $P_0 = P_q$ and $q \geq 3$ by the hypothesis, there exist distinct indices $j_1,j_2$ such that $s_{j_1} = s_{j_2+1}$ and for every $j \in \{j_1,j_1+1,\ldots,j_{2}\}$ the $s_j$'s are pairwise different (indices are taken modulo $q$).
Then, two cases may arise:
\begin{itemize}
\item Case $j_2 = j_1 + 1$.
We have: 
$$e_{j_1} := v_{i_{j_1}}v_{i_{j_1}+1}, \ e_{j_2} := v_{i_{j_2}+1}v_{i_{j_2}}  \in E(s_{j_1},s_{j_2}).$$
Furthermore, $C_0$ is clearly a shortest cycle passing by $v_{i_{j_1}}$ (by $v_{i_{j_1}+1},v_{i_{j_2}+1},v_{i_{j_2}}$, respectively), and so, by Property~\ref{claim:step-1}, either we computed a short cycle of weight $\leq 2w(C_0)$ during Step 1, or we have $w(C_0) \geq 2 \cdot \max\{dist(s_{j_1},v_{i_{j_1}}),dist(s_{j_1},v_{i_{j_2}+1}),dist(s_{j_2},v_{i_{j_1}+1}),dist(s_{j_2},v_{i_{j_2}})\}$.
In the latter case, there exists a cycle of weight: 
$$\leq dist(s_{j_1},v_{i_{j_1}}) + w_{e_{j_1}} + dist(s_{j_2},v_{i_{j_1}+1}) + dist(s_{j_2},v_{i_{j_2}}) + w_{e_{j_2}} + dist(s_{j_1},v_{i_{j_2}+1}) \leq 3w(C_0)$$ that is fully contained in $V(T_{s_{j_1}}) \cup V(T_{s_{j_2}})$.
By Property~\ref{claim:step-4}, we so computed a cycle of weight $\leq 9 w(C_0)$ at Step 4.
\item From now on let us assume $j_2 \neq j_1 + 1$. For every $j$ we have $e_{j}:= v_{i_j}v_{i_j+1} \in E(s_j,s_{j+1})$, and so, the edge $s_js_{j+1} \in E_S$ has weight no more than $dist(s_j,v_{i_j}) + w_{e_j} + dist(v_{i_j+1},s_{j+1})$ in $H_S$ (indices are taken modulo $q$ for the $s_j$'s and modulo $p$ for the $v_i$'s).
Furthermore, $C_0$ is clearly a shortest cycle passing by $v_{i_j}$ (by $v_{i_j+1}$, respectively), and so, by Property~\ref{claim:step-1}, either we computed a short cycle of weight $\leq 2w(C_0)$ during Step 1, or we have $w(C_0) \geq 2 \cdot \max\{dist(s_j,v_{i_j}), dist(s_{j+1},v_{i_j+1}\}$ for every $j$.
In the latter case, $(s_{j_1},s_{j_1+1},\ldots,s_{j_2},s_{j_2+1}=s_{j_1})$ is a cycle in $H_S$ of weight: 
\begin{align*}
&\leq w(C_0) + \sum_{j=j_1}^{j_2} (dist(v_{i_{j-1}+1},s_j) + dist(s_j,v_{i_j})) \\ \\
&\leq w(C_0)(1+(j_2-j_1+1)) \\ \\
&\leq w(C_0)|S| \\ \\
&= \tilde{\cal O}(n^{2/3} \cdot w(C_0)).
\end{align*}
Then, let $C_H$ be a cycle of $H_S$ such that $w(C_H) = \tilde{\cal O}(n^{2/3} \cdot w(C_0))$ (obtained by applying the algorithm of Roditty and Tov to $H_S$).
As explained above, we can derive from $C_H$ a cycle $C_H^*$ of $G$ such that $w(C_H^*) \leq w(C_H) = \tilde{\cal O}(n^{2/3} \cdot w(C_0))$.
\end{itemize} 
Summarizing, we obtain an $\tilde{\cal O}(n^{2/3})$-approximation of the girth by outputting a shortest cycle computed during Steps 1,3,4,5.
\end{proof}

\subsection{Improving the approximation factor}

We can now conclude this section with its main result:

\mainUnbounded*

\begin{proof}
We may assume that all weights are positive by Lemma~\ref{remove-zero-weights}.
Let $g^*$ be the $\tilde{\cal O}(n^{2/3})$-approximation that we compute using Proposition~\ref{prop:linear-approx}.
There exists some (known) constant $c$ such that the girth of $G$ is somewhere between $g^*/(cn^{2/3}\log{n})$ and $g^*$.
Now, let $i_{\min},i_{\max}$ be the smallest nonnegative integers such that $g^*/(cn^{2/3}\log{n}) \leq (1+\varepsilon/2)^{i_{\min}}$ and in the same way $g^* \leq (1+\varepsilon/2)^{i_{\max}}$.
We have $i_{\max} - i_{\min} = {\cal O}\left(\log_{1+\varepsilon/2}{\left( \frac{g^*}{g^*/(cn^{2/3}\log{n})} \right)}\right) = {\cal O}(\log{n}/\log{(1+\varepsilon/2)}) = {\cal O}(\log{n}/\varepsilon)$. 

\smallskip
Let $S$ be as in Proposition~\ref{prop:good-set}.
For every $v \in V \setminus S$, we compute a $2$-approximation of a shortest cycle in $G'_v$: the subgraph of $G$ induced by the ball $B_S(v)$. -- In fact, this is already done in the proof of Proposition~\ref{prop:linear-approx}, but we restate it here for completeness of the method. --
By Theorem~\ref{thm:previous-approx}, it can be done in time $\tilde{\cal O}(n^{2/3})$ for each $v$, and so, this takes total time $\tilde{\cal O}(n^{5/3})$.

Then, let $T = \{ (1+\varepsilon/2)^i \mid i_{\min} \leq i \leq i_{\max} \}$.
For every $s \in S$, we compute the smallest $t \in T$ such that $HBD(G,s,t)$ detects a cycle (if any). 
It can be done in time $\tilde{\cal O}(|S|n\log{|T|}) = \tilde{\cal O}(n^{5/3}\log{1/\varepsilon})$ by using a dichotomic search.

\medskip
Let $g_{\min}$ be the value computed by the above algorithm (with a corresponding cycle).
We claim that the girth of $G$ is at least $g_{\min}/(2+\varepsilon)$, that will prove the theorem. 
For that, let us consider a shortest cycle $C_0$.
We need to consider several cases:
\begin{itemize}
\item Assume that $C_0$ passes by some $s \in S$.
Let $i_0$ be the smallest index such that $w(C_0) \leq (1+\varepsilon/2)^{i_0}$.
Since we have $g^*/(cn^{2/3}\log{n}) \leq w(C_0) \leq g^*$, $i_0 \in \{i_{\min},i_{\min}+1,\ldots,i_{\max}\}$.
Furthermore, by Corollary~\ref{cor:stop-hbd-1} applied for $dist(s,C_0) = 0$, $HBD(G,s,(1+\varepsilon/2)^{i_0})$ detects a cycle.
By Lemma~\ref{lem:stop-output}, we so deduce that $g_{\min} \leq 2(1+\varepsilon/2)^{i_0} \leq 2(1+\varepsilon/2)w(C_0) \leq (2+\varepsilon)w(C_0)$.
\item Otherwise, let $v \in V \setminus S$ be contained into $C_0$.
We may assume that $V(C_0) \not\subseteq B_S(v)$, since otherwise $C_0$ is a cycle in $G_v'$, and so, we can already conclude that $g_{\min} \leq 2w(C_0)$.
In particular, we have: $$\max_{u \in V(C_0)} dist_{C_0}(u,v) \geq \max_{u \in V(C_0)} dist_{G}(u,v) \geq dist_G(v,S) > 0.$$
Let $s \in S$ be such that $dist_G(v,s) = dist_G(v,S)$.
By Corollary~\ref{cor:stop-hbd-2}, this implies that the smallest $t_s$ such that $HBD(G,s,t_s)$ detects a cycle satisfies $t_s \leq w(C_0)$.
Therefore, we can choose as above the smallest index $i_0$ such that $w(C_0) \leq (1+\varepsilon/2)^{i_0}$.
As already noticed $i_0 \in \{i_{\min},i_{\min}+1,\ldots,i_{\max}\}$, and by Corollary~\ref{cor:stop-hbd-2} we know that $HBD(G,s,(1+\varepsilon/2)^{i_0})$ detects a cycle.
In this situation, we can conclude by Lemma~\ref{lem:stop-output} that $g_{\min} \leq 2(1+\varepsilon/2)^{i_0} \leq 2(1+\varepsilon/2)w(C_0) \leq (2+\varepsilon)w(C_0)$.  
\end{itemize}
Overall, the above proves as claimed $g_{\min} \leq (2+\varepsilon)w(C_0)$.
\end{proof}

\section{A subquadratic algorithm for dense graphs}\label{sec:remove-edges}

A drawback of the algorithms in Theorems~\ref{thm:main-bounded-weight} and~\ref{thm:main-unbounded-weight} is that their time complexity also depends on the number $m$ of edges.
It implies that for dense graphs with $m = \Theta(n^2)$ edges we do not achieve any improvement on the running time compared to~\cite{LiL09,RoT13}.
Assuming sorted adjacency lists, we now prove that the dependency on $m$ can always be discarded:

\mainWithoutM

Theorem~\ref{thm:main-without-m} is obtained from a natural modification of the algorithm that we presented in Section~\ref{sec:main}.
We postpone its proof to Section~\ref{sec:subquadratic-in-n-proof}.
Section~\ref{sec:unweighted-density} is devoted to a different approach for the problem: that combines the result of Theorem~\ref{thm:main-unbounded-weight} with a well-known density result for {\em unweighted} graphs in order to derive either a {\em randomized} $4$-approximation, or a deterministic $c$-approximation for some constant $c > 4$.

\subsection{From weighted to unweighted graphs and back}\label{sec:unweighted-density}

We first prove the following result:

\begin{proposition}\label{prop:unweighted}
For every $G=(V,E,w)$ with non negative edge-weights and sorted adjacency lists, we can compute:
\begin{enumerate}
\item a randomized $4$-approximation for {\sc Girth} in expected time $\tilde{\cal O}(n^{5/3})$;
\item and, for every $\varepsilon > 0$, a deterministic $(8+\varepsilon)$-approximation for {\sc Girth} in time $\tilde{\cal O}(n^{5/3}\polylog{1/\varepsilon})$.
\end{enumerate}
\end{proposition}

Informally we will use Proposition~\ref{prop:unweighted} as a replacement for Proposition~\ref{prop:linear-approx} in the proof of Theorem~\ref{thm:main-without-m}.
We think that our approach in this section could also be of independent interest.
Our result combines Theorem~\ref{thm:main-unbounded-weight} with the following well-known result in graph theory:

\begin{theorem}[~\cite{Bro66}]\label{thm:c4-threshold}
Every {\em unweighted} graph with order $n$ and $m \geq \left( \frac 1 2 + o(1)\right) n^{3/2}$ edges contains a cycle of length four.
\end{theorem}

\begin{algorithm}[!ht]
        \texttt{Controlled-Density-Approx}$(G)$
       
         \vspace{1mm}

        {\bf Assumption:} the adjacency lists $(Q_v)_{v \in V}$ are sorted.
        \vspace{1mm}
        \begin{algorithmic}[1]\small
            \STATE $H = (V,\emptyset,w), w_{\max} \gets 0$
            \STATE
            \STATE {\it /* Initialization */}
            \STATE $Q \gets \{\}$
            \FORALL{$v \in V$}
            \STATE $uv \gets {\tt Extract-min}(Q_v)$
            \STATE $Q \gets Q \cup \{uv\}$
            \ENDFOR
            \STATE
            \STATE {\it /* Construction of the subgraph */}
            \WHILE{$Q \neq \emptyset$ {\bf and} $|E(H)| < \left( \frac 1 2 + o(1)\right) n^{3/2}$}
            \STATE $uv \gets {\tt Extract-min}(Q)$
            \STATE $E(H) \gets E(H) \cup \{uv\}$; $w_{\max} \gets w_{uv}$
            \STATE $uv' \gets {\tt Extract-min}(Q_u)$; $u'v \gets {\tt Extract-min}(Q_v)$
            \STATE $Q \gets Q \cup \{uv', u'v \}$
            \ENDWHILE
            \STATE
            \STATE Let $C$ be s.t. $w(C)$ is a $(2 + \varepsilon/4)$-approximation of the girth of $H$.
            \RETURN $C$.
        \end{algorithmic}
\end{algorithm}

\begin{proofof}{Proposition~\ref{prop:unweighted}}
If $G$ has $m = \tilde{\cal O}(n^{5/3})$ edges then, we can simply apply Theorem~\ref{thm:main-unbounded-weight} for $\varepsilon = 2$ (the latter can be easily verified by scanning the adjacency lists until we read the end of it or we reach the desired upper-bound).
From now on assume this is not the case and let $H$ be induced by the $\left\lceil\left( \frac 1 2 + o(1)\right) n^{3/2}\right\rceil$ edges of minimum weight in $G$.

We claim that $H$ can be constructed in time $\tilde{\cal O}(n^{3/2})$ by using a priority queue $Q$.
Indeed, initially we set $E(H) = \emptyset$ and for every $v \in V$ we start inserting in $Q$ the edge of minimum-weight that is incident to $v$.
This way, we ensure that a minimum-weight edge of $G \setminus E(H)$ is present in $Q$ (recall that initially, $E(H) = \emptyset$, and so, $G = G \setminus E(H)$).
Then, in order to preserve this above invariant, each time a minimum-weight edge $uv$ is extracted from $Q$ and added in $H$ we insert in $Q$ the remaining edge of minimum weight in $Q_u$ and the one in $Q_v$ (if any).
-- Note that in doing so, a same edge can be added in $Q$ twice, but this has no consequence on the algorithm. --

We now apply Theorem~\ref{thm:main-unbounded-weight} for $\varepsilon' = \varepsilon/4$ to $H$, and we so obtain a cycle $C$ that is a $(2 + \varepsilon/4)$-approximation of the girth of $H$.
We claim that $w(C)$ is also a $(8+\varepsilon)$-approximation of the girth of $G$.
In order to prove this claim, we need to consider two different cases:
\begin{itemize}
\item Assume there exists a shortest cycle $C_0$ of $G$ such that $E(C_0) \subseteq E(H)$.
By Theorem~\ref{thm:main-unbounded-weight}, $w(C) \leq (2+\varepsilon/4)w(C_0) < (8+\varepsilon)w(C_0)$.
\item Otherwise, any shortest cycle $C_0$ of $G$ has at least one edge that is not contained in $H$.
Since edges are added by increasing weights, this implies that every such a shortest cycle contains an edge of weight at least $w_{\max}$.
In particular, the girth of $G$ is at least $w_{\max}$, where $w_{\max}$ denotes the maximum-weight of an edge in $H$.
Furthermore, since $H$ has enough edges by construction, by Theorem~\ref{thm:c4-threshold} it contains a cycle of four vertices; the latter has weight at most $4w_{\max}$.
As a result, $w(C) \leq (2+\varepsilon/4) \cdot 4w_{\max} = (8 + \varepsilon)w_{\max} \leq (8+\varepsilon)w(C_0)$.
\end{itemize}
The above proves the claim, and so, the deterministic version of the result.
In order to obtain a randomized $4$-approximation, it suffices to pick $\varepsilon \leq 2$ and to output any cycle $C'$ of $H$ with four vertices (then, we output any of $C,C'$ that has minimum weight).
Up to some constant multiplicative increase of the number of edges to add in $H$, this can be done by using a randomized algorithm of Yuster and Zwick that runs in expected linear time~\cite[Theorem $2.9$]{YuZ97}. 
Note that this is the only source of randomness in the algorithm.
\end{proofof}

\subsection{Derandomization}\label{sec:subquadratic-in-n-proof}

We end up derandomizing the result of Proposition~\ref{prop:unweighted}.

\begin{proofof}{Theorem~\ref{thm:main-without-m}}
It suffices to prove the result for the graphs with non negative edge-weights and $\varepsilon >0$ arbitrary.
To see that, let us consider any graph $G$ with all edge-weights in $\{1,2,\ldots,M\}$.
If we take $G$ as the input of our deterministic algorithm for the graphs with non negative edge-weights, setting $\varepsilon = \frac 1 {M(n+1)}$, then, the cycle $C$ of $G$ that is outputted by the algorithm has total integer weight strictly upper-bounded by $1/\varepsilon$.
In particular, $w(C)$ is in fact a $4$-approximation of the girth of $G$.
Thus, from now on we will only consider the more general case of graphs $G$ with non negative edge-weights.

\smallskip
Let $\varepsilon > 0$ be fixed.
The randomized version of the theorem was already proved in Proposition~\ref{prop:unweighted}.
For proving the deterministic version, we start computing a $c$-approximation $g^*$ of the girth of $G$, for some universal constant $c > 4$.
By Proposition~\ref{prop:unweighted}, it can be done in time $\tilde{\cal O}(n^{5/3})$.
Observe that the girth of $G$ is somewhere between $g^*/c$ and $g^*$.
Now, define $\eta = \varepsilon/6$.
Let $i_{\min},i_{\max}$ be the smallest nonnegative integers such that $g^*/c \leq (1+\eta)^{i_{\min}}$ and in the same way $5g^* \leq (1+\eta)^{i_{\max}}$.
We have $i_{\max} - i_{\min} = {\cal O}(1/\log{(1+\eta)}) = {\cal O}(1/\varepsilon)$. 
Furthermore, let $T = \{ (1+\eta)^i \mid i_{\min} \leq i \leq i_{\max} \}$.

The remaining of the algorithm is essentially the same as for Theorem~\ref{thm:main-unbounded-weight}, except that we avoid the costly computation of the induced subgraphs $G'_v$.
Specifically, our processing of the vertices in the set $S$ (given by Proposition~\ref{prop:good-set}) remains unchanged.
However, for every $v \in V \setminus S$, we compute the smallest $t \in T$ such that $t < dist_G(v,S)$ and $HBD(G,s,t)$ detects a cycle (if any).
In doing so, we can only visit the vertices in $B_S(v)$, and so, the total running time for processing $v$ is upper-bounded by $\tilde{\cal O}(|B_S(v)|\polylog{M})$, that is in $\tilde{\cal O}(n^{1/3}\polylog{M})$.
Overall, this algorithm runs in total time $\tilde{\cal O}(n^{5/3}\log{T})$, that is in $\tilde{\cal O}(n^{5/3}\log{1/\varepsilon})$. 

Let $C \in \{ C_v \mid v \in V \}$ be of minimum weight amongst all the cycles computed by the algorithm.
We claim that $w(C)$ is a $(4+\varepsilon)$-approximation of the girth of $G$.
In order to prove this claim, let $C_0$ be a shortest cycle of $G$.
There are two cases.
\begin{itemize}
\item Assume there exists a vertex $s \in V(C_0) \cap S$.
Then, $w(C) \leq w(C_s) \leq 2(1+\eta) w(C_0) < (4+\varepsilon)w(C_0)$ (this is a similar proof as for Theorem~\ref{thm:main-unbounded-weight}).
\item Thus, from now on we assume $V(C_0) \cap S = \emptyset$.
Let $v \in V(C_0)$ be arbitrary.
\begin{itemize}
\item Assume first $w(C_0) < dist_G(v,S)/(1+\eta)$.
By Corollary~\ref{cor:stop-hbd-1} applied for $dist_G(v,C_0) = 0$, the smallest $t_v$ such that {\tt HBD}$(G,v,t_v)$ detects a cycle satisfies $t_v \leq w(C_0)$.
In particular, the smallest $t_v' \in T$ such that {\tt HBD}$(G,v,t_v')$ detects a cycle satisfies $t_v' \leq (1+\eta)w(C_0) < dist_G(v,S)$.
By Lemma~\ref{lem:stop-output}, $w(C) \leq w(C_v) \leq 2(1+\eta)w(C_0) < (4+\varepsilon)w(C_0)$.
\item Otherwise $w(C_0) \geq dist_G(v,S)/(1+\eta)$.
Let $s \in S$ minimize $dist_G(s,v)$.
Then, by Corollary~\ref{cor:stop-hbd-1}, the smallest $t_s$ such that ${\tt HBD}(G,s,t_s)$ detects a cycle satisfies: $$t_s \leq dist_G(s,v) + w(C_0) = dist_G(v,S) + w(C_0) \leq (2+\eta)w(C_0).$$
In particular, the smallest $t_s' \in T$ such that ${\tt HBD}(G,s,t_s')$ detects a cycle satisfies: $$t_s' \leq (1+\eta)t_s \leq (2+2\eta + \eta^2)w(C_0) \leq (2+3\eta)w(C_0).$$
As a result, by Lemma~\ref{lem:stop-output} $w(C) \leq w(C_s) \leq (4+6\eta)w(C_0) \leq (4+\varepsilon)w(C_0)$.
\end{itemize}
\end{itemize}
Summarizing, $w(C) \leq (4+\varepsilon)w(C_0)$ in all the cases.
\end{proofof}

We leave open whether a better approximation-factor than $4$ can be obtained in $o(n^2)$-time.

\section{Open problems}\label{sec:ccl}

The most pressing question is whether we can achieve a $4/3$-approximation for the girth in subquadratic time.
If it is not the case then, what is the best approximation factor we can get in subquadratic time?
We note that in~\cite{RoV12}, Roditty and Vassilevska Williams conjectured that we cannot achieve a $(2-\varepsilon)$-approximation already for unweighted graphs.
If true then, this would imply our algorithm is essentially optimal (at least for the non dense graphs with ${\cal O}(n^{2-\varepsilon})$ edges).

\bibliographystyle{abbrv}
\bibliography{bib-girth}

\end{document}